\documentclass[envcountsame]{llncs}
\usepackage[utf8]{inputenc}

\usepackage{url}
\usepackage{amsmath}
\usepackage{amssymb}
\usepackage{amsfonts}
\usepackage{tikz}
\usepackage{appendix}
\tikzstyle{every node} = [draw, fill=white, circle, inner sep=0pt, minimum size=5pt]
\tikzstyle{d} = [very thick]
\tikzstyle{n} = [draw=none, rectangle, inner sep=6pt] 
\tikzstyle{i} = [draw, fill=black, circle, inner sep=0pt, minimum size=5pt] 

\usepackage{fancyhdr}
\newcommand{\N}{\mathcal{N}}
\newcommand{\A}{\mathcal{A}}
\newcommand{\F}{\mathcal{F}}

\newcommand{\simm}{{\sim}}
\newcommand{\hatt}{\hat{\hspace{0.2cm}}}
\newcommand{\m}{\mathbf}
\newcommand{\Wk}{\mathcal W}
\newcommand{\RWkRA}{\mathsf{RWkRA}}
\newcommand{\RwkRA}{\mathsf{RwkRA}}
\newcommand{\wkRA}{\mathsf{wkRA}}
\newcommand{\RRA}{\mathsf{RRA}}
\newcommand{\RA}{\mathsf{RA}}

\spnewtheorem{definition}{Definition}{\bfseries}{\rmfamily}

\title{Representable and diagonally representable weakening relation algebras\thanks{This work was supported by the Engineering and Physical Sciences Research Council EP/S021566/1}}
\author{Peter Jipsen \inst{1}\orcidID{0000-0001-8608-808X} \and Ja{\v s} {\v S}emrl  \inst{2}  \orcidID{0000-0001-7440-8867}}
\institute{Chapman University \email{jipsen@chapman.edu} \url{https://www1.chapman.edu/~jipsen/}
\and UCL (University College London) \email{j.semrl@cs.ucl.ac.uk} \url{http://www0.cs.ucl.ac.uk/staff/jsemrl/}}

\begin{document}

\maketitle

\begin{abstract}
    A binary relation defined on a poset is a weakening relation if the partial order acts as a both-sided compositional identity. This is motivated by the weakening rule in sequent calculi and closely related to models of relevance logic. For a fixed poset the collection of weakening relations is a subreduct of the full relation algebra on the underlying set of the poset. We present a two-player game for the class of representable weakening relation algebras akin to that for the class of representable relation algebras. This enables us to define classes of abstract weakening relation algebras that approximate the quasivariety of representable weakening relation algebras. We give explicit finite axiomatisations for some of these classes. We define the class of diagonally representable weakening relation algebras and prove that it is a discriminator variety. We also provide explicit representations for several small weakening relation algebras.

\end{abstract}

\section{Introduction}
The \emph{full algebra of binary relations on} $X$ is 
\[
\mathbf{Rel}(X)=(\mathcal{P}(X^{2}),\cap,\cup,\emptyset,\top,;,id_{X},\neg,^{\smallsmile})
\]
where $\top=X^{2}$,
$R{;}S$ is the composition of $R$, $S$ , $\neg R=X^{2}\setminus R$, and
$R^{\smallsmile}=\{(x,y)\mid(y,x)\in R\}$.
The class \emph{\textsf{RRA}} of \emph{representable relation algebras} $=\mathbb{SP}\{\mathbf{Rel}(X)\mid X$
is a set$\}$.
Tarski \cite{tarski1956contributions} proved that $\RRA$ is a variety and Monk \cite{Mon1964}
proved that $\RRA$ is not finitely axiomatisable. For more details see the books by Givant \cite{Giv2017a}, \cite{Giv2017} and Maddux \cite{Mad2006}.

The set of \emph{weakening relations} on a poset $\m X=(X,\le)$ is $\Wk(\m X)=\{R\subseteq X^{2}\mid{\le}{;}R{;}{\le}=R\}$.
The \emph{full algebra of weakening relations on a poset} $\m X$
is $$\mathbf{wk}(\m X)=(\Wk(X,\le),\cap,\cup,\emptyset,\top,;,1,\sim)$$
where $1={\le}$ and ${\sim}R=\neg R^{\smallsmile}$ is the complement-converse operation.
The class of \textbf{representable weakening relation algebras} is
$$\RwkRA=\mathbb{SP}\{\mathbf{wk}(X,\le)\mid(X,\le) \mbox{ is a
poset}\}.$$ 

Weakening relations are the analogue of binary relations
when the category \textbf{Set} of sets and functions is replaced by
the category \textbf{Pos} of partially ordered sets and order-preserving
functions.
Since sets can be considered as discrete posets (i.e. antichains, ordered
by the identity relation), \textbf{Pos} contains \textbf{Set} as a
full subcategory, which implies that weakening relations are a substantial
generalisation of binary relations.
However, weakening relations do not allow $\neg$ or $^{\smallsmile}$
as operations.

They have applications in sequent calculi \cite{GJ2017}, 
quasi-proximity lattices/spaces \cite{smyth1992stable}
, order-enriched categories \cite{kurz2016relation}
, mathematical
morphology \cite{Ste2015}, and program semantics, e.g. via separation logic \cite{reynolds2002separation}
.

The closely related algebras $\m{Wk}(\m X)$ are defined as the expansions of $\m{wk}(\m X)$ by the Heyting implication $R\to S=\{(x,y)\mid \forall u, v(u \le x\ \&\ y \le v \ \&\ u R v \Rightarrow u S v)\}$. The $\mathbb{SP}$-closure of these algebras is denoted by $\RWkRA$ and has been studied in \cite{GJ2020}, \cite{GJ2020a}, \cite{Jip2017}, \cite{Ste2012}, \cite{Ste2015}. It is a discriminator variety that has $\RRA$ of representable relation algebras as a proper subvariety. The algebras in $\RWkRA$ are generalised bunched implication algebras, and the algebras in $\RwkRA$ are all the subreducts of algebras in $\RWkRA$, hence $\RwkRA$ is a quasivariety. We show that it is not a variety, but with respect to representability the two classes behave the same way.

In Section 2 we define a representation game for $\RwkRA$ (which can be extended to a game for $\RWkRA$) and use it to give an explicit universal axiomatisation for the class. Section 3 defines (Kripke) frames for weakening relation algebras and adapts the game to this setting. From an $n$-pebble version of this frame game we define a sequence of classes $\wkRA_n$ that approximate $\RwkRA$ from above, similar to the sequence $\RA_n$ that converges the $\RRA$. In the next section we find finite axiomatisation for $\wkRA_2$ and $\wkRA_3$. In Section 5 we define the class of representable diagonal weakening relation algebras and show that is a discriminator variety. Finally, in the last section we show that all associative algebras in $\wkRA_3$ with 6 elements or fewer are representable.

\section{Representation game}

In this section we present a representation game for weakening relation algebras similar to those defined for relation algebras, defined in \cite{hirsch2002relation}. We begin by defining some notation.

\begin{definition}
    A \emph{bounded cyclic involutive unital distributive lattice-ordered magma}  $\A = (A,\cdot,+,\bot,\top,;,1,\sim)$ is an algebra such that
    \begin{enumerate}
        \item $(A,\cdot, +,\bot,\top)$ is a bounded distributive lattice
        \item $(s+t);(u + v) = s;u + s;v + t;u + t;v$
        \item $s{;}\bot=\bot=\bot{;}s$
        \item $s{;}1 = s=1{;}s$
        \item $\simm(\simm s) = s$
        \item $\simm(s \cdot t)=\simm s + \simm t$
    \end{enumerate}
     for all $s,t,u,v \in A$. A \emph{representation} of $\A$ is an injective homomorphism $h:\A\to \m{wk}(\m X)$ for some poset $\m X=(X,\le)$ such that $h(\top)$ is an equivalence relation on $X$.
\end{definition}

    Note that $s \leq t$ if and only if $s + t = t$, or equivalently $\simm s \cdot \simm t = \simm t$ which can be rewritten as $\simm t \leq \simm s$, hence $\sim$ is order reversing. The adjective ``cyclic'' is included in the name to contrast it to the non-cyclic general case the are two unary operations $\sim,-$ in the language that satisfy $\simm{-}s=s={-}\simm s$. In the cyclic case $\simm, {-}$ have the same interpretation.

Distributive lattice-ordered magmas are abbreviated as $d\ell$-magmas.
Let $\A$ be a bounded cyclic involutive unital $d\ell$-magma. Additionally we define $0 = \simm 1$.

\begin{definition}
    A \emph{network (for $\A$)} is a tuple $\N = (N, \lambda)$ where $N$ is a set of \emph{nodes} and $\lambda: N^2 \to \wp(\A)$ is a \emph{labelling function} such that for all $x,y \in N$, $1 \in \lambda(x,x)$ and $\top \in \lambda(x,y)$.
    Such a network 
    is \emph{consistent} if and only if for all $x,y \in N$ we have that
    $$\lambda(x,y) \cap \{\simm a \mid a \in \lambda(y,x)\} = \emptyset.$$
    A network $\N = (N, \lambda)$ is a \emph{prenetwork} of $\N' = (N', \lambda')$ -- denoted $\N \subseteq \N'$ -- if and only if $N \subseteq N'$ and for all $x,y \in N$ we have $\lambda(x,y) \subseteq \lambda'(x,y)$.
\end{definition}

    Observe that the prenetwork predicate is a partial order and that inconsistency is inherited from prenetworks.

We now have the tools to define a two player game and prove that the existence of a winning strategy for one of the players coincides with $\A$'s membership in the class of $\RWkRA$.

\begin{definition}
    An \emph{$n$-round representation game}, denoted $\Gamma_n(\A)$, for some $n \leq \omega$ is a two player game played between the challenger $\forall$ (Abelard) and the responder $\exists$ (H\'{e}lo\"{\i}se) over $n+1$ moves. After the $i$th move for $0 \leq i \leq n$, $\exists$ will return a network $\N_i$ such that $\N_0 \subseteq \N_1 \subseteq ... \subseteq \N_n$. The game is won by $\forall$ if $\exists$ returns an inconsistent network. Otherwise $\exists$ wins.

    On the \emph{initialisation move} $\forall$ picks a pair of elements $a \nleq b \in \A$ and $\exists$ must return a network $\N_0$ with  some $(x,y) \in N_0^2$ such that $a \in  \lambda(x,y)$ and $\simm b \in  \lambda(y,x)$.

    On the $i$th move for $0 < i \leq n$, $\forall$ may challenge $\exists$ with any of the following four moves.

    \begin{description}
        \item[join move:] $\forall$ picks $x,y \in N_{i-1}$, some $a \in {\lambda}_{i-1}(x,y)$, and some $b,c \in \A$ such that $a \leq b+c$. $\exists$ must return a $\N_i$ with $b \in {\lambda}_i(x,y)$ or $c \in {\lambda}_i(x,y)$.
        \item[involution  move:] $\forall$ picks $x,y \in N_{i-1}$ and some $a,b \in \A$ such that $b = \simm a$.\linebreak $\exists$ must return a $\N_i$ with $a \in {\lambda}_i(x,y)$ or $b \in {\lambda}_i(y,x)$.
        \item[composition move:] $\forall$ picks $x,y,z \in N_{i-1}$ and $a \in {\lambda}_{i-1}(x,y), b \in {\lambda}_{i-1}(y,z)$. $\exists$ must return a $\N_i$ with $ c \in {\lambda}_i(x,z)$ where $c = a;b$.
        \item[witness move:] $\forall$ picks $x,y \in N_{i-1}$, $a \in {\lambda}_{i-1}(x,y)$, and $b,c \in \A$ such that $a = b;c$. $\exists$ must return a $\N_i$ with some $z \in N_i$ such that $ b \in {\lambda}_i(x,z), c \in {\lambda}_i(z,y)$.
    \end{description}
\end{definition}

\begin{proposition}\label{prop:repgame}
    $\A$ is representable if and only if $\exists$ has a winning strategy for $\Gamma_\omega(\A)$.
\end{proposition}

\begin{proof}
    If $\A$ is representable, then $\exists$ can take some representation $h$ over $X$. Let $a \nleq b$ be the pair played on initialisation move. There will exist some maximal $X' \subseteq X$ such that $\exists x,y \in X': (x,y) \in h(a)\setminus h(b)$ and $\forall z,w \in X': (z,w) \in h(\top)$. On initialisation move, $\exists$ can return the network $\N = (X', \lambda)$ where $\lambda(x,y) = \{c \in \A \mid (x,y) \in h(c)\}$. Because $h$ preserves all the operations in the language, all moves $\forall$ may call are trivially responded to by returning the same network after every move.

    If $\A$ is countable then $\forall$ can schedule his moves in a way that every move will be called eventually. Let $\N^{a,b}_0, \N^{a,b}_1, \N^{a,b}_2, ...$ be the networks during an $\exists$-winning play of $\Gamma_\omega(\A)$ where $\forall$ scheduled his moves in such a way and the initialisation move was called for the pair $a \nleq b$. Define $N^{a,b}_\omega$ as $\{x \mid \exists i < \omega: j \geq i \Rightarrow x \in N^{a,b}_i\}$, $\lambda^{a,b}_\omega(x,y)$ as $\{c \mid \exists i < \omega: j \geq i \Rightarrow (x,y \in N^{a,b}_j \wedge c \in \lambda^{a,b}_j(x,y))\}$, and a relation $\equiv$ as $\{(x,y) \in (N^{a,b}_i)^2 \mid 1 \in \lambda^{a,b}_\omega(x,y), 1 \in \lambda^{a,b}_\omega(y,x) \}$. It is symmetric by definition, reflexive because networks are defined as having $1 \in \lambda^{a,b}_\omega(x,x)$ and transitive because all composition moves were called eventually and $1;1=1$. Therefore, we can define $h^{a,b}: \A \to \big((N_\omega/{\equiv})^2\big)$ where for all $c \in \A$ we have $h^{a,b}(c) = \{([x]_\equiv,[y]_\equiv) \mid c \in \lambda^{a,b}_\omega(x,y)\}$.

    Because of initialisation there exists a pair $x,y \in N^{a,b}_\omega$ with $a \in \lambda^{a,b}_\omega(x,y)$ and as the network remains consistent and $\simm b \in \lambda^{a,b}_\omega(y,x)$, we have $b \notin \lambda^{a,b}_\omega(x,y)$. Because every composition move was called and $1$ is the identity for $;$, we have $b \notin \lambda^{a,b}_\omega(x',y')$ for all $x' \in [x]_\equiv, y' \in [y]_\equiv$. Thus $([x]_\equiv, [y]_\equiv)$ is a pair that ensures that $h^{a,b}(a) \not \subseteq h^{a,b}(b)$.

    $\top$ is represented by $h^{a,b}$ as the top relation by the definition of a network and $\bot$ is represented as the empty relation because if there was a pair $(x,y)$ with $\bot \in \lambda_\omega^{a,b}(x,y)$ then by a series of composition moves every pair would include $\bot$ in its label. That would mean that the initialisation pair of points would include both $a$ and $b$ in its label.

    The join move ensures that join is represented correctly by $h^{a,b}$, namely the join move ensures $h^{a,b}(c)+h^{a,b}(d) \subseteq h^{a,b}(c+d)$ because $c \leq (c+d) + (c+d)$ and $h^{a,b}(c)+h^{a,b}(d) \supseteq h^{a,b}(c+d)$ because $c+d \leq c+d$ and thus a pair in $h^{a,b}(c+d)$ will also be in $h^{a,b}(c)$ or $h^{a,b}(d)$.

    Because an inconsistent network is never introduced we have that if $\simm c \in \lambda^{a,b}_\omega(x,y)$ then $c \not\in \lambda^{a,b}_\omega(x,y)$. Because all composition moves are eventually called and $1$ is the identity for $;$ that applies to the all the points in the relevant equivalence classes of $\equiv$ and $h^{a,b}(\simm c) \subseteq \simm h^{a,b}(c)$. For $h^{a,b}(\simm c) \supseteq \simm h^{a,b}(c)$ if $([x]_\equiv,[y]_\equiv) \notin h^{a,b}(\simm c)$ then when the involution move was called for $(y,x)$ and $c$, $\exists$ chose to have $c \in \lambda^{a,b}_\omega(y,x)$.

    If $([x]_\equiv,[y]_\equiv) \in h^{a,b}(c)$ and $([y]_\equiv,[z]_\equiv) \in h^{a,b}(d)$ then by composition moves of $a;1;b = a;b$ we have $([x]_\equiv,[z]_\equiv) \in h^{a,b}(c;d)$ so $h^{a,b}(c);h^{a,b}(d) \subseteq h^{a,b}(c;d)$. For $h^{a,b}(c);h^{a,b}(d) \supseteq h^{a,b}(c;d)$, assume $([x]_\equiv,[z]_\equiv) \in h^{a,b}(c;d)$ and without loss $c;d \in \lambda_\omega^{a,b}(x,z)$. Because the relevant witness move was called there will exist a $y$ such that $([x]_\equiv,[y]_\equiv) \in h^{a,b}(c)$ and $([y]_\equiv,[z]_\equiv) \in h^{a,b}(d)$.

    $h^{a,b}(1)$ is reflexive by the definition of a network, antisymmetric as the base was quotiented by $\equiv$, and transitive as $1;1=1$ and all the composition moves were called eventually. Furthermore as the composition is represented by $h^{a,b}$ and $1$ being the identity for composition, it is the case that for all $c \in \A$, $h^{a,b}(c)$ is a weakening relation with respect to $1$.

    $h^{a,b}$ is thus a homomorphism for $\A$ discriminating $a \nleq b$. Thus let $h(c)$ for all $c \in \A$ be the disjoint union $\dot\bigcup_{a \nleq b \in \A}h^{a,b}(c)$. Because $h$ is a homomorphism that discriminates all $a \nleq b$ pairs, it is a representation.

    This generalises to uncountable algebras by the downward L{\" o}wenheim Skolem Theorem since $\RWkRA$ is a pseudoelementary class.
\qed\end{proof}

Next we show that the existence of a winning strategy for $\exists$ can be expressed by a universal first-order sentence. For this result we define the following concepts.

\begin{definition}
    A \emph{term network} is a network $\N = (N,\lambda)$ where $N$ is a finite set of nodes and $\lambda$ is a labelling function that maps every pair of nodes to a finite set of terms. We also require that for all $x,y \in N$, $1 \in \lambda(x,x)$ and $\top \in \lambda(x,y)$.

    For every term network $\N = (N,\lambda)$ we define a network $\N^{+,x,y,t}=(N \cup \{y\}, \lambda^\ell)$ where $x \in N, y \in N \uplus \{x^+\}$ (for some new node $x^+$), $t$ is a term in the language of $\RWkRA$ and for all $z,w\in N \uplus \{x^+\}$. 
$$\lambda^\ell(z,w) = \begin{cases}
        \{1,\top\}  & \text{if } x^+=z=w\\
        \{\top\}  &\text{if }x^+=z \neq w\ne x\mbox{ or } x^+=w\ne z \neq x, w \neq y\\
\{t, \top\} &\text{if } z=x,w=y=x^+\\
        \lambda(z,w) \cup \{t\} &\text{if } z=x,w=y\ne x^+\\
        \lambda(z,w) & \mbox{otherwise}
    \end{cases}$$

    For variables $a,b$ we define two initial term networks below.
    \begin{align*}
        \N^{1,a,b} = & (\{x\}, \{(x,x) \mapsto \{\top,1,a,\simm b\} \})\\
        \N^{2,a,b} = & (\{x,y\}, \{(x,x) \mapsto \{\top,1\},(x,y) \mapsto \{\top,a\}, \\
        &(y,x) \mapsto \{\top,\simm b\}, (y,y) \mapsto \{\top,1 \} \})
    \end{align*}
\end{definition}

\begin{proposition}
    \label{prop:fmlasgame}
    For every $n < \omega$ there exists a first-order formula $\sigma_n$ that corresponds to $\exists$ having a winning strategy for $\Gamma_n(\A)$.
\end{proposition}

\begin{proof}

    We show by induction that there exists a formula $\phi_n(\N)$ for every $0 \leq n < \omega$, defined for a finite term network $\N$, with all the variables universally quantified that signifies that the network can remain consistent for $n$ more moves of the representation game where $\exists$ plays \emph{conservatively}, i.e., only adds the requested labels. It is easy to see that she has a winning strategy for the game if and only if she also has one for the conservative play.

    In the base case, $\phi_0(\N)$ defined below signifies consistency (remaining consistent for zero moves)
    $$\phi_0(\N) = \bigwedge_{x,y \in N} \bigwedge_{t \in \lambda(x,y)}\bigwedge_{t' \in \lambda(y,x)} t \neq \simm t'.$$

    In the induction case, we assume that $\phi_{n}(\N')$, where $\N'$ is a term network with all variables universally quantified, is both necessary and sufficient for $\N'$ to be able to remain consistent for $n$ moves. Then we show you can define $\phi_{n+1}(\N)$ that extends the assumption to $n+1$ moves. Although we use $a,b$ here, the variable names should be unique when constructing these formulas.
     \begin{align*}
         \phi_{n+1}(\N) = & \bigwedge_{x,y \in N} \bigwedge_{t \in \lambda(x,y)} \forall a,b \big( t \leq a + b \implies (\phi_n(\N^{+,x,y,a}) \vee \phi_n(\N^{+,x,y,b})) \big)\\
         \wedge & \bigwedge_{x,y \in N} \bigwedge_{t \in \lambda(x,y)} \forall a \big( \phi_n(\N^{+,x,y,a}) \vee \phi_n(\N^{+,y,x,\simm a}) \big)\\
         \wedge & \bigwedge_{x,y,z \in N} \bigwedge_{t \in \lambda(x,y)} \bigwedge_{t' \in \lambda(y,z)} \phi_n(\N^{+,x,z,t;t'}) \\
         \wedge & \bigwedge_{x,y \in N} \bigwedge_{t \in \lambda(x,y)} \forall a,b \Big( t = a;b \implies \!\!\!\!\bigvee_{z \in N \uplus \{x^+\}} \phi_n\big((\N^{+,x,z,a})^{+,z,y,b} \big)\Big)
     \end{align*}

     We now have a formula $\phi_n(\N)$ for every $0 \leq n < \omega$ that ensures $\exists$ can keep a universally quantified term network $\N$ consistent. Hence the formula 
     $$\sigma_n = \forall a,b\ (a \nleq b \implies (\phi_n(\N^{1,a,b}) \vee \phi_n(\N^{2,a,b}))$$
 ensures that $\exists$ has a winning strategy for a conservative game of length $n$.
 \qed\end{proof}

\begin{corollary}
    \label{cor:fmlasgame}
    $\Sigma = \{\sigma_1, \sigma_2, \ldots\}$ together with the axioms for cyclic distributive involutive semirings is a recursively enumerable theory that axiomatises $\RWkRA$.
\end{corollary}

\section{Frames, frame games, and finite pebble games}

In this section we present finite algebras as frames, similar to Routley-Meyer frames or relevance frames for relevance logic \cite{BDM2009} and atom structures of atomic relation algebras \cite{Mad1982}. We then define a modified version of the representation game that utilises frames.

Finally, we define an $n$-pebble versions of the frame game. Analogous to the abstract classes of relation algebras $\RA_\omega \subseteq \hdots \subseteq \RA_3 \subseteq \RA_2$, this gives rise to classes of weakening relation algebras $\wkRA_\omega \subseteq \hdots \subseteq \wkRA_3 \subseteq \wkRA_2$. Clearly $\RA_\omega, \wkRA_\omega$ are the classes of representable relation algebras and weakening algebras, respectively. Furthermore, similarly to $\RA_4$, we say that $\wkRA_4$ is the class of weakening relation algebras.

First, observe that the language of $\RwkRA$ does not include negation and hence the lattice need not be Boolean. As we will see in Section~\ref{smallRwkRA}, the smallest representable non-Boolean algebra is a 4-element chain $\m S_4$. Thus we cannot present finite weakening relation algebras using atoms. Instead, we make use of join-irreducibles.

\begin{definition}
    A non-$\bot$ element $a$ of a representable weakening relation algebra is \emph{join-irreducible} if and only if for all $b,c$ if $a = b+c$ then $a=b$ or $a=c$.
It is \emph{join-prime} if and only if for all $b,c$ if $a \leq b+c$ then $a \leq b$ or $a \leq c$.
\end{definition}

Because $\cdot$ distributes over $+$ we have that an element is join-irreducible if and only if it is join-prime. In the finite case every algebra will have join-irreducibles and every element is a join of join-irreducibles. (In general this is only true for \emph{perfect} algebras. In fact, by definition, a distributive lattice is \emph{join-perfect} if every element is a join of completely join-irreducible elements. This generalises the concept of \emph{atomic} for Boolean algebras.)

The element $a$ in the result below is called the \emph{join-irreducible label} of $(x,y)$.
\begin{proposition}
    \label{prop:labelirr}
    In a representation $h$ of a finite representable weakening relation algebra $\A$, for any pair $(x,y)$ there exists a join-irreducible $a \in \A$ such that
    $${\uparrow}a=\{s \in \A \mid (x,y) \in h(s) \}.$$
\end{proposition}

\begin{proof}
    A representation $h$ maps joins to unions, hence the set $\{s \mid (x,y) \in h(s) \}$ is upward closed and if $(x,y) \in h(a + b)$ then it is also in $h(a)$ or $h(b)$. Hence the base set of the representation is itself a union of upward closures of join-primes. Now if it is above ${\uparrow}a$ and ${\uparrow}b$ then it must be the case that $(y,x)$ is in neither $h(\simm a)$ nor $h(\simm b)$ and thus $(x,y) \in h(\simm(\simm a + \simm b)) = h(a \cdot b)$. Thus the meet of all such join-irreducibles must also be a non-$\bot$ element that is join-prime and below all elements in the set.
\qed\end{proof}

Although the converse operation is not defined in our language, we can use the following trick to define a useful unary operation on the join-irreducibles.

\begin{definition}
    For every join-irreducible $a$ in a finite algebra, define $\hat{a} = \simm\sum_{a \nleq s} s$ where $\sum$ is with respect to join ($+$).
\end{definition}

    The join $\sum_{s \nleq t} t$, defined for all $s$ in a finite algebra $\A$, is usually denoted $\kappa(s)$. If we take $s \leq s' \in \A$ we have $s \nleq t \Rightarrow s' \nleq t$ and thus $\kappa(s) \leq \kappa(s')$, hence $\kappa$ is order preserving. Because $\simm$ is order reversing and $\kappa$ is order preserving we have that $\hatt$ is order reversing. 

\begin{proposition}
    \label{prop:jitildaji}
    In any finite bounded distributive involutive additive algebra $\A$, if $a$ is a join-irreducible, so is $\hat{a}$. 
\end{proposition}

\begin{proof}
    It is well known that $\kappa(a)$ of a join-irreducible $a$ in a lattice is meet irreducible and because $\sim$ is order reversing, that means that $\hat{a} = \simm \kappa(a)$ is a join-irreducible. 
\qed\end{proof}

\begin{proposition}
    If a pair $(x,y)$ in a representation has the join-irreducible label $a$, then $(y,x)$ has label $\hat{a}$. Moreover, $\hat{\hat a}=a$.
\end{proposition}

\begin{proof}
   $\simm s \in h(y,x)$ if and only if $s \notin h(x,y)$, i.e. $a \nleq s$. Thus, by the argument from Proposition~\ref{prop:labelirr} the join-irreducible label of $(y,x)$ can be written as $\prod_{a \nleq s} \simm s$ where $\prod$ is with respect to meet ($\cdot$) and this is equivalent to $\simm \sum_{a \nleq s} s$ by the De Morgan equivalence.
\qed\end{proof}

Finally to characterise composition, we need to define a ternary predicate, similar to the set of allowed triangles in relation algebras.





\begin{definition} Let $\A$ be a finite bounded cyclic involutive $d\ell$-magma and
    define a ternary relation $R$ on the set of join-irreducibles of $\A$ by $$R(a,b,c)\text{ if and only if }
    a \leq b;c.$$
\end{definition}

For relation algebras with atoms $a,b,c$ the Peircian triangle law says that
$$
         a \leq b;c \iff  \hat{a} \leq \hat{c};\hat{b}  \iff b \leq a;\hat{c}
\iff        \hat{b} \leq c;\hat{a} \iff c \leq \hat{b};a \iff\hat{c} \leq \hat{a};b.
$$
As we will see in the next section, this law does not hold for the class of representable weakening relation algebra frames. However, atom structures for relation algebras generalise to the weakening setting as follows. 


\begin{definition}
    \label{def:frame}
    A \emph{relevance frame} $\mathcal{F} = (F, I, \leq, R, \hatt)$ is a structure with a carrier set $F$, a unary predicate $I$, a partial order predicate $\leq$, a ternary predicate $R$, and an order-reversing involution operation $\hatt$ where for all $a,b,c,d$ in $F$
    \begin{enumerate}
        \item \label{it:lel} $a \leq b \Leftrightarrow \exists e: I(e) \wedge R(a,e,b)$
        \item \label{it:ler} $a \leq b \Leftrightarrow \exists e: I(e) \wedge R(a,b,e)$
        \item \label{it:dncls} $a \leq b \wedge R(b,c,d) \Rightarrow R(a,c,d)$
        \item \label{it:ucls} $b \leq c \wedge R(a,b,d) \Rightarrow R(a,c,d)$
        \item \label{it:uucls} $c \leq d \wedge R(a,b,c) \Rightarrow R(a,b,d)$
    \end{enumerate}    
\end{definition}

\begin{proposition}
    A relevance frame $\mathcal{F} = (F, I, \leq, R, \hatt)$ defines a bounded involutive unital $d\ell$-magma $\A = (A,\cdot,+,\bot,\top,;,1,\sim)$ by taking $(F, \leq)$ as the join-irreducibles of the lattice with their partial order and for all $s,t \in A$
    $$1 = \sum_{I(a)} a, \hspace{1cm}
        \simm s = \sum_{\hat{a} \nleq s} a,\hspace{1cm}
        s;t = \sum_{b \leq s, c \leq t, R(a,b,c)} a$$
    where $a,b,c\in F$.    
\end{proposition}

\begin{proof}
    A bounded distributive lattice can be defined by its join-irreducibles and their ordering. To show that the magma is unital, we can see that no term of the join defining the composition with the identity is above the identity by Definition~\ref{def:frame}(\ref{it:lel})(\ref{it:ler}) and because $\leq$ is reflexive, there will exist, for every join-irreducible a term in the composition with the identity (on either side) equal to that join-irreducible. Thus $1$ is precisely the identity. Composition is additive by definition. $\sim$ is an involution because a join-irreducible $a \leq \simm (\simm s)$ if and only if $\hat{a} \nleq \simm s$ which is true if and only if $a = \hat{\hat{a}} \leq s$. For the De Morgan equivalence, $a \leq \simm(\simm s + \simm t)$ if and only if $\hat{a} \nleq \simm s + \simm t$, or equivalently $\hat{a} \nleq \simm s \wedge \hat{a} \nleq \simm t$ which by definition is true if and only if $a=\hat{\hat{a}} \leq s$ and $a=\hat{\hat{a}} \leq t$, or simply $a \leq s \cdot t$.
\qed\end{proof}

\begin{proposition}
    Every finite bounded cyclic involutive unital $d\ell$-magma has a unique equivalent relevance frame.
\end{proposition}

\begin{proof}
    Finite distributive lattices are determined by their poset of join-irreducibles, and from Proposition~\ref{prop:jitildaji} they have a unique $\hatt$ defined on the join-irreducibles. The mapping to $R$ is unique as Definition~\ref{def:frame}(\ref{it:dncls})(\ref{it:ucls})(\ref{it:uucls}) ensure that $R$ is downward closed in the first argument and upward closed in the other arguments.
\qed\end{proof}

\begin{proposition}
    For finite algebras and finite frames, the mappings described in the previous two lemmas are inverses of each other.
\end{proposition}

\begin{proof}
    Finite distributive lattices correspond uniquely to their posets of join-irreducibles. The preservation of identity and the composition follow trivially from the definition. For $\simm,\hatt$ observe that $ \simm \kappa(a) = \hat{a} \nleq s$ if and only if $\simm s \nleq \kappa(a) = \sum_{a \nleq} t$, or equivalently $a \leq \simm s$. For the converse note that $\simm \kappa(a) = \sum_{\hat{b} \nleq \kappa(a)}b = \sum_{a \leq \hat{b}}b = \sum_{b \leq \hat{a}}b = \hat{a}$.
\qed\end{proof}

Although we have defined these frames for finite algebras, we can say that a possibly infinite algebra is \emph{frame-definable} if it can be defined by a relevance frame. In the context of relation algebras, this corresponds to complete and atomic relation algebras. Similarly to that class, we will show that every non-frame definable algebra embeds into a frame definable algebra with equivalent representability.

\begin{proposition}
    Let $\A$ be a bounded cyclic involutive unital $d\ell$-magma. Then a frame $\mathcal{F}(\A)$ can be defined by taking the carrier set of all prime filters $U \subseteq A$, with $U \leq V$ if and only if $V \subseteq U$, $\hat{U} = \{\simm s \mid s \in A \setminus U \}$, $I(U) \Leftrightarrow 1 \in U$ and $R(U,V,W)$ if and only if for all $v \in V, w \in W$ we have 
    $v {;}w \in U$.
\end{proposition}

\begin{proof}
    $\leq$ is clearly a partial order, for closure of $\hatt$ note that $A \setminus U$ is a prime ideal, so by the order reversing property of $\sim$, $\hat{U}$ is a prime filter. Furthermore, all the unitality conditions are trivially preserved and by $U' \leq U$ if and only if $U \subseteq U'$ we have downward closure of $R$ in the first argument and upward closure in the other two.
\qed\end{proof}


\begin{proposition}
$\A$ is representable if and only if the algebra defined by $\F(\A)$ is a representable weakening relation algebra. This algebra is called the \emph{canonical extension} of $\A$.
\end{proposition}

\begin{proof}
    Because $\A$ is a subalgebra of the algebra defined by $\F(\A)$, we know that the right to left implication is true. For the other direction, if $\A$ is representable, every $(x,y)$ will have a prime filter $U$ such that $(x,y) \in h(a)$ if and only if $a \in U$ to represent the lattice correctly. The prime filter defining $(y,x)$ will be exactly $\hat{U}$. The identity is also correctly represented as it is only above those prime filters that include it. Finally for $;$ we have shown that it suffices for $\exists$ to have a winning strategy for a game of any finite length. Thus at any point we need to show that the compositions are correctly represented if and only if all compositions finite meets are properly included in the relevant prime filter.
\qed\end{proof}

\begin{definition}
    A \emph{frame network} $\N = (N,\lambda)$ is defined for a frame $\mathcal{F} = (F, I, \leq, R, \hatt)$ with $N$ being the set of nodes and $\lambda: N^2 \to F$ is the labelling function. The network is said to be \emph{consistent} if and only if for all $x,y \in N$ we have $\lambda(x,y) = \widehat{\lambda(y,x)}$ and for all $x,y,z \in N$ we have $R(\lambda(x,y), \lambda(x,z), \lambda(z,y))$.

    We say for two frame networks $\N = (N, \lambda),\; \N' = (N', \lambda')$ that $\N \subseteq \N'$ if and only if $N \subseteq N'$ and $\lambda = \lambda'\restriction_{N^2}$ where $\restriction$ denotes the restriction of the function to the domain in the subscript.
\end{definition}

\begin{definition}
    An infinite length \emph{frame game} $G(\mathcal{F})$ where $\mathcal{F} = (F, I, \leq, R, \hatt)$ is a relevance frame is defined for two players $\forall$ and $\exists$.

    The game starts with $\forall$ picking a join-irreducible $a$ and $\exists$ must return a frame network $\N_0 = (N_0, \lambda_0)$ such that there exists $x,y \in N_0$ such that $\lambda_0(x,y) = a$.

    At the $i$th move for $0 < i < \omega$ $\forall$ picks a pair $x,y \in N_{i-1}$ and a pair of join-irreducibles $a,b$ such that $R(\lambda(x,y), a, b)$ and for all $a' \leq a, b' \leq b \in N_{i-1}$ if $R(\lambda(x,y), a', b')$ then $a=a'$ and $b=b'$. $\exists$ must return a network $\N_i = (N_i, \lambda_i)$ such that $\N_{i-1} \subseteq \N_i$ and $\exists z \in N_i$ such that $\lambda(x,z) = a, \lambda(z,y) = b$.

    $\forall$ wins if and only if $\exists$ returns an inconsistent network at any point in the game.
\end{definition}

\begin{proposition}
    $\exists$ has a winning strategy for $G(\mathcal{F}(\A))$ if and only if she has a winning strategy for $\Gamma_\omega(\A)$.
\end{proposition}

\begin{proof}
    It suffices to prove that she has a winning strategy for the play where all moves are called eventually. Thus if she has a winning strategy for $\Gamma_\omega(\A)$, we know that the limit network will have the relevant prime filters as labels. Thus if $a$ is the initial join-irreducible $\exists$ can map all her moves from the limit network of the play where the initialisation pair was $a \nleq \kappa(a)$.

    For the converse, assume she has a winning strategy for $G(\mathcal{F}(\A))$. To respond to the initialisation move with $s \nleq t$ there will exist a join-irreducible $a$ such that $a \leq s$ but $a \nleq t$ or rather $t \leq \kappa(a)$ so returning the initial network for $a$ will ensure that $a \leq s$ and $\hat{a} \leq \tilde t$. Any witness move called can be responded to by minimal join-irreducible pairs, which makes any other witness moves called by $\forall$ redundant.
\qed\end{proof}

We now define for every $2 \leq n \leq \omega$ the $n$-pebble equivalent version of the frame game as follows.

\begin{definition}
    The $n$-pebble infinite move game $G^n(\F)$ for a frame $\F$ is defined exactly as $G(\F)$, except before $\forall$ calls a witness move, he takes $\N' \subseteq \N_{i-1}$ such that $|N'| \leq n$ and then proceeds to call the witness move.
\end{definition}

    In particular, the frame game $G$ is equivalent to $G^\omega$.
Next we define $\wkRA_n$ and $\wkRA$ analogous to  $\RA_n$, the variety of all $n$-dimensional relation algebras, and $\RA$, the variety of all (4-dimensional) relation algebras \cite{Mad1983}, \cite{hirsch2002relation}.

\begin{definition}
    The class $\wkRA_n$ is the class of all bounded cyclic involutive unital $d\ell$-magmas $\A$ for which $\exists$ has a winning strategy for $G^n(\F(\A))$.
    The class of weakening relation algebras $\wkRA$ is defined as $\wkRA_4$.
\end{definition}

   It follows that $\wkRA_\omega$ is equivalent to $\RwkRA$
    and $\wkRA_\omega \subseteq \hdots \subseteq \wkRA_4 \subseteq \wkRA_3 \subseteq \wkRA_2$.

\section{Axiomatisation of the abstract classes}

In this section we provide finite axiomatisations for $\wkRA_2$ and $\wkRA_3$. We leave open the problem of whether, similarly to $\RA_4$ the axiomatisation for $\wkRA_4$ consists of axioms of $\wkRA_3$ and associativity of $;$.

We begin by axiomatising $\wkRA_2$. This will be done using the axiomatisation of bounded cyclic involutive unital $d\ell$-magmas together with the theory $\Phi_2$, defined below.

\begin{definition}
    Let $\Phi_2$ be the first order theory given by the following quasiequations:
    \begin{enumerate}
        \item $s \cdot \simm s \leq 0$
        \item $s \leq t \Leftrightarrow s;\simm t \cdot 1 \leq 0 $
        \item $s \leq t{;}u \wedge s{;}t \leq \simm u\Rightarrow s \cdot 1 \leq 0$
        \item $s \leq t{;}u \wedge u{;}s \leq \simm t \Rightarrow s \cdot 1 \leq 0$
        \item $s \leq t{;}u \wedge (s \cdot 1 \cdot t{;}v) + (1 \cdot s \cdot \simm v{;}u) \leq 0 \Rightarrow s \cdot 1 \leq 0$
    \end{enumerate}
\end{definition}

Before we prove the soundness and completeness, we introduce a ternary predicate for the language of frames $R^{\min}$ from the equivalence below.
$$R^{\min}(a,b,c) \Leftrightarrow R(a,b,c) \wedge \forall b',c': (R(a,b',c') \wedge b' \leq b \wedge c' \leq c \Rightarrow b'=b \wedge c'=c )$$
Note that since the union of a chain of prime filters is again a prime filter, frames of the form $\F(\A)$ have the property that $R(a,b,c)$ can be refined to $R^{\min}(a,b',c')$ for some prime filters $b'\le b$ and $c'\le c$.

\begin{lemma}
    \label{lem:phi2frame}
    Let $\A$ be a bounded cyclic involutive unital $d\ell$-magma. $\A \models \Phi_2$ if and only if $\F(\A)$ satisfies
    \begin{enumerate}
        \item \label{it:w2id} $\forall a \exists b: I(b) \wedge \hat{b}=b \wedge R(b,a, \hat{a})$
        \item \label{it:w2r1} $\forall a,b: I(a) \wedge \hat{a}=a \wedge R(a,b,\hat{b}) \Rightarrow R(b,a,b)$
        \item \label{it:w2r2} $\forall a,b: I(a) \wedge \hat{a}=a \wedge R(a,b,\hat{b}) \Rightarrow R(\hat{b},\hat{b},a)$
        \item \label{it:w2r3} $\forall a,b,c: I(a) \wedge \hat{a}=a \wedge R^{\min}(a,b,c) \Rightarrow b = \hat{c}$
    \end{enumerate}
\end{lemma}

\begin{proof}
    For the left to right implication, observe that for any join-irreducible $a$, we know that $a \nleq \kappa(a)$ so $(a{;}~\simm \kappa(a) \cdot 1) \nleq 0$ by $\Phi_2(2)$. Thus there must exist a join-irreducible $b \leq 1, b \nleq 0, b \leq a{;}\hat{a}$. Suppose $b \neq \hat{b}$. Then there would exist some $b \leq s, \hat{b} \nleq s$. Because $\hat{\hat{b}} = b$ we know that $b \leq \simm s$ and thus $b \leq s \cdot \simm s \leq 0$, contradicting $\Phi_2(1)$ and we have proven (\ref{it:w2id}) follows from $\Phi_2$. For (\ref{it:w2r1}) assume we have $a \leq 1, \hat{a} = a$ then $a \nleq \simm 1 = 0$. Thus $a = a \cdot 1 \nleq 0$ and $a \leq b{;}\simm \kappa(b)$ implies $a{;}b \nleq \simm \simm \kappa(b)$ or simply $a \leq a{;}b$. By a similar argument we get (\ref{it:w2r2}). Finally if $a \leq b{;}c$ and $a = a \cdot 1 \leq b{;}c \cdot 1 \nleq 0$ we have $b \nleq \simm c$. We also have that $a \cdot 1 = a \nleq 0$ and $a \leq b{;}c$ so $a \cdot \kappa(b){;}c \nleq 0$ or $a \cdot b{;}\hat{b} \nleq 0$. In the former case that means that $\kappa(b){;}c \nleq 0$ and thus $\kappa(b) \leq \simm c$ or $c \leq \hat{b}$ and we are done. In the latter case it means that there exists a join-irreducible $a' \leq a$ such that $a \nleq 0$ and thus $a' = \hat{a'}$ as well as $a' \leq b;\hat{b} \leq b;c$ by monotinicity. Because $a' \leq a \leq \hat{a'} = a'$ we have $a=a'$ and by minimality $\hat{b}=c$.

    For the right to left implication note that if $s \cdot \simm s \nleq 0$ then there exists some $a \leq s \cdot \simm s$ not below $0$. Thus $\hat{a} \leq 1$ and we know there exists a join-irreducible $b = \hat{b}, b \leq a;\hat{a} \leq (s \cdot \simm s){;}1 = s \cdot \simm s$ and that contradicts $b = \hat{b}$. Assume $s \nleq t$. That is true if there exists a join-irreducible $a$ such that $a \leq s, \hat{a} \leq \simm t$. Thus there exists a join-irreducible $b = b \cdot 1 \nleq 0$ below $1 \cdot a;\hat{a} \leq 1 \cdot s{;}\simm t$ and we conclude $1 \cdot s{;}\simm t \nleq 0$. If $1 \cdot s{;}\simm t \nleq 0$ then there exist join-irreducibles $a,b,c$ such that $I(a), a = \hat{a}, b \leq s, c \leq \simm t$ and $b,c$ also being minimal and hence $\hat{b} = c$. Therefore $\hat{b} \leq \simm t$ or simply $s \nleq t$. $s \cdot 1 \nleq 0 \wedge s \leq t{;}u \Rightarrow s{;}t \nleq \simm u$ follows directly from (\ref{it:w2r1}) and its dual directly from (\ref{it:w2r2}). Finally $s \cdot 1 \nleq 0$ and $s \leq t;u$ iff there exist some $a,b,c$ in the corresponding frame such that $a \leq s \cdot 1, b \leq t, c \leq u, I(a), \hat{a}=a, R^{\min}(a,b,c)$ and thus $\hat{b} = c$ by (\ref{it:w2r3}). Observe that for every $v$ either $b \leq v$ or $\hat{b} \leq \simm v$ and thus $a \leq t;v$ or $a \leq \simm v{;}u$ and the join of the two terms is not below $0$. 
\qed\end{proof}

\begin{theorem}
    $\wkRA_2$ is axiomatised by the basic axioms for bounded cyclic involutive unital $d\ell$-magmas and $\Phi_2$.
\end{theorem}

\begin{proof}
    By Lemma~\ref{lem:phi2frame} this axiomatisation is equivalent to the frame conditions, enumerated (\ref{it:w2id})--(\ref{it:w2r3}). First we show these are sound for the two pebble game. If there existed a join-irreducible $a$ with no $b, I(b), \hat{b}=b$ with $R(b,a,\hat{a})$, then $\forall$ would win on initialisation with $a$ because if $\lambda(x,y) = a$, no consistent $b$ would exist for $\lambda(x,x)$. We show (\ref{it:w2r3}) next. If this didn't hold for some $a,b,c$ then $\forall$ could start by asking $a$ on initial move. By order reversing of $\hatt$ and the identity, $a$ is the only join-irreducible to be set as $\lambda(x,x)$ where $\lambda(x,y)=a$. For the second move, $\forall$ calls the witness $b;c$ on $(x,x)$ and we get an inconsistency because $b \neq \hat{c}$. For (\ref{it:w2r1}) and (\ref{it:w2r2}) see that if $R(a,b,\hat{b})$ we have $R^{\min}(a,b,\hat{b})$ by order reversing properties of $\hatt$ and (\ref{it:w2r3}). Thus if $\forall$ again starts by forcing $\lambda(x,x)=a$ then calling the witness $b;\hat{b}$ then both $R(b,a,b), R(\hat{b},\hat{b},a)$ must hold to keep the network consistent.

    To show completeness, it suffices to say that $\exists$ can respond to any initialisation with $a$ by returning a network with two nodes $x,y$ with $\lambda(x,y) = a, \lambda(y,x) = \hat{a}$ and by (\ref{it:w2id}) there exists a $b$ for $a$ and $b'$ for $\hat{a}$ to be set as $\lambda(x,x)$ and $\lambda(y,y)$ respectively and by (\ref{it:w2r1})(\ref{it:w2r2}) all other triangles are also consistent. A witness move can only be called on a reflexive node $(x,x)$ and that means that by (\ref{it:w2r1})(\ref{it:w2r2})(\ref{it:w2r3}) any witness will be consistent and by the same reasoning as with initialisation, $\exists$ can put a label on $\lambda(y,y)$ and keep the network consistent.
\qed\end{proof}

In order to axiomatise $\wkRA_3$ we only need to add two well known axioms as well as a set of quasiequations. The first axiom is called \emph{rotation} for involutive semirings 
and the second one was found by Maddux in \cite{Mad2022} as an axiom that holds for binary relations, but not for relevance logic frames.

\begin{definition}
    Let $\Phi_3$ be the first order theory containing all the formulas in $\Phi_2$ as well as 
    \begin{enumerate}
        \item $s;t \leq \simm u \Rightarrow t;u \leq \simm s$
        \item $s \cdot t{;}u \leq ((s{;}v) \cdot t){;}u + t{;}(u \cdot \simm v)$
        \item $1 \cdot \simm s'{;}s \cdot t{;}\simm t' \leq 0 \Rightarrow s{;}t \leq (s \cdot s'){;}t + s{;}(t \cdot t')$
        \item $1 \cdot s \cdot 0 = \bot \Rightarrow (s\cdot 1){;}(t{;}u) \leq ((s \cdot 1){;}t){;}u $
        \item $1 \cdot u \cdot 0 = \bot \Rightarrow (s{;}t){;}(u \cdot 1) \leq s{;}(t{;}(u \cdot 1))$
    \end{enumerate}
\end{definition}

\begin{lemma}
    \label{lem:phi3frame}
    Let $\A$ be a bounded cyclic involutive unital $d\ell$-magma. $\A \models \Phi_3$ if and only if for $\F(\A)$ all the formulas from Lemma~\ref{lem:phi2frame} hold as well as
        \begin{enumerate}
        \item \label{it:wk3r1} $\forall a,b,c: R^{\min}(a,b,c) \Rightarrow R(b,a,\hat{c})$
        \item \label{it:wk3r3} $\forall a,b,c:  R(a, \hat{b}, \hat{c}) \Rightarrow R(b, \hat{c}, \hat{a})$
        \item \label{it:wk3i1} $\forall a,b,c: R^{\min}(a,b,c) \Rightarrow \exists d : d = \hat{d} \wedge I(d) \wedge R(d,\hat{b},b) \wedge R(d, c, \hat{c})$
        \item \label{it:wk3i2} $\forall a,b,c,d: d = \hat{d} \wedge I(d) \wedge R(a,d,a) \wedge R^{\min}(a,b,c) \Rightarrow R(b,d,b)$
        \item \label{it:wk3i3} $\forall a,b,c,d: d = \hat{d} \wedge I(d) \wedge R(a,a,d) \wedge R^{\min}(a,b,c) \Rightarrow R(c,c,d)$
    \end{enumerate}
\end{lemma}

\begin{proof}
For the left to right implication of (\ref{it:wk3r1}) if $a,b,c$ are join-irreducibles with $a \leq b{;}c$ as well as the minimality condition for $b,c$ then see that $a = a \cdot b{;}c \leq (a{;}\hat{c} \cdot b){;}c + b{;}(c \cdot \kappa(c))$.  $c \cdot \kappa(c)$ is strictly below $c$ and due to minimality of $b,c$ for this composition $a \nleq b{;}(c \cdot \kappa(c))$. Thus $a \leq (a{;}\hat{c} \cdot b){;}c$ and again by minimality $a{;}\hat{c} \cdot b = b$ or simply $R(b,a,\hat{c})$. 
For (\ref{it:wk3r3}) observe that $a \nleq \hat{b};\hat{c}$ is the same as $\simm \kappa(b){;} \simm \kappa(c) \leq \kappa(a)$ and by rotate we get $\simm \kappa(c){;}\simm \kappa(a) \leq \kappa(b)$ and $\simm \kappa(a){;}\simm \kappa(b) \leq \kappa(c)$ so $R(a, \hat{b}, \hat{c}), R(b, \hat{c}, \hat{a}), R(c, \hat{a}, \hat{b})$  are equivalent.
For (\ref{it:wk3i1}) if $R^{\min}(a,b,c)$ then we know $b{;}c \nleq (b \cdot \kappa(b)){;}c + b{;}(c \cdot \kappa(b))$ and thus $1 \cdot \simm s'{;}s \cdot t{;}\simm t' \nleq 0$ and we can find a $d$ satisfying $I(d), \hat{d} = d, R(d,\hat{b},b), R(d, c, \hat{c})$. For (\ref{it:wk3i2}) see that $1 \cdot d \cdot 0 = \bot$ and thus $a \leq d{;}a \leq d{;}(b{;}c) \leq (d{;}b){;}c$. By minimality $b = d{;}b$. By a similar argument we get (\ref{it:wk3i3}).

For the right to left implication, if $s{;}t \leq \simm u$ observe that for all join-irreducibles $a,b,c$ such that $a \leq s, b \leq t, c \leq u$ we have $a{;}b \leq \kappa(\hat{c})$ and thus $\neg R(\hat{c},a,b)$ and by (\ref{it:wk3r3}) we have $\neg R(\hat{a},b,c)$ and thus $b{;}c \leq \kappa(\hat{a}) = \simm a$. If for all join-irreducibles $a,b,c$ below $s,t,u$ respectively that holds then $t{;}u \leq \simm s$. To show $s \cdot t{;}u \leq ((s{;}v) \cdot t){;}u + t{;}(u \cdot \simm v)$ take any $a \leq s \cdot t{;}u$ and some minimal $b,c$ witnessing the $t{;}u$ composition. Then all $v$ will either have $c \leq \simm v$ or $\hat{c} \leq v$, in either case the term is above $a$ by monotonicity.
Finally if $s{;}t \nleq (s \cdot s'){;}t \cdot s{;}(t \cdot t')$ it means that $s{;}t$ is non-empty and as such there exists some $a \leq s{;}t$ and some $R^{\min}(a,b,c)$ and as such $b \nleq s'$ and $c \nleq t'$ and thus $\hat{b};b \leq \simm s'{;}s$ and  $c{;}\hat{c} \leq t{;}\simm t'$ and there exists a $d \nleq 0$ such that $d \leq 1 \cdot \simm s'{;}s \cdot t{;}\simm t'$ and therefore the term cannot be below $0$. Take any join-irreducible $a \leq (s \cdot 1);t{;}u$. There will exist a self-$\hatt$ join-irreducible $d \leq s \cdot 1$ such that $d \leq d;a$ and a minimal $b,c$ below $t,u$ such that $a \leq b;c$ and so we have by (\ref{it:wk3i2}) $b \leq d;b$ and thus $a \leq b{;}c \leq (d{;}b){;}c \leq ((s \cdot 1){;}t){;}u$. The dual is shown similarly from  (\ref{it:wk3i3}).
\qed\end{proof}

\begin{theorem}\label{theo:wk2}
    $\wkRA_3$ is axiomatised by the basic axioms for bounded cyclic involutive unital $d\ell$-magmas and $\Phi_3$.
\end{theorem}

\begin{proof}
    First we show that all the formulas from Lemma~\ref{lem:phi3frame} are sound. If we have $a,b,c$ such that $R^{\min}(a,b,c)$ then $\forall$ calls $a$ on initialisation and calls the witness $R^{\min}(a,b,c)$ on the $\lambda(x,y)=a$ and $\exists$ must return such a network where $\lambda(x,z) = a, \lambda(y,z) = \hat{c}$ so $R(b,a,\hat{c})$ must hold for consistency and we have (\ref{it:wk3r1}). For (\ref{it:wk3r3}) assume without loss that we have $R(a,\hat{b}, \hat{c})$ so there must be some minimal $\hat{b'} \leq \hat{b}, \hat{c'} \leq \hat{c}$ to call the witness on the initial pair $a$. Observe that for consistency $b \leq b' \leq \hat{c'};\hat{a} \leq \hat{c};\hat{a}$ by monotonicity.
    For (\ref{it:wk3i1}) if $\forall$ initialises with $a$ and calls the $b,c$ witness, $\exists$ needs a join-irreducible $d$ to put on the reflexive edge of the added node.
    
    From Lemma~\ref{lem:phi2frame}, Theorem~\ref{theo:wk2} we have that $\exists$ can survive the initial move and we only need to examine the two possible witness moves, that on a non-reflexive edge in a two-node network and that on a reflexive edge. If a witness move $R^{min}(a,b,c)$ is called on a non-reflexive edge $(x,y)$, check that all Peircian transformations of this triangle hold. By (\ref{it:wk3r1}) we have $R(b,a,\hat{c})$ and through (\ref{it:wk3r3}) we get $R(\hat{b},c,\hat{a}), R(\hat{c},\hat{a},b)$ from $R(a,b,c)$ and $R(\hat{a},\hat{c},\hat{b}), R(c,\hat{b},a)$ from $R(b,a,\hat{c})$. For the reflexive edge on $(z,z)$ you can see that $\exists$ can add $\lambda(z,z) = d$ from(\ref{it:wk3i1}) and by similar reasoning to Theorem~\ref{theo:wk2} all triangles including $(z,z)$ are consistent. Finally let $\lambda(x,x) = d$. By (\ref{it:wk3i2}) $R(b,d,b)$ and by (\ref{it:wk3r3}) $R(\hat{d}=d,b,\hat{b})$. The consistency of other triangles follows from formulas in Lemma~\ref{lem:phi2frame}. Similarly we get consistency for $\lambda(y,y)$. For the reflexive witness $R^{min}(d,a,\hat{a})$ on $(x,x)$ observe due to order reversing of $\hatt$, $\exists$ can either find a join-irreducible $c$ such that $R^{\min}(\lambda(x,y),a,c)$ or $R^{\min}(\lambda(y,x), c, \hat{a})$ and $\exists$ can use the same strategy as for the non-reflexive witness move.  
\qed\end{proof}

To axiomatise the class $\wkRA = \wkRA_4$ we would at least need to add associativity for composition. For $\RA$, it is precisely the axioms for $\RA_3$ and composition that axiomatise $\RA_4$, however, whether this also holds for $\wkRA$ remains open.

\begin{problem}
    What axioms are necessary to axiomatise $\wkRA$? Is it finitely axiomatisable?
\end{problem}

\begin{problem}
    Let $n > 4$. $\RA_n$ is not finitely axiomatisable \cite{hirsch2002relation}. Is the same true for $\wkRA_n$?
\end{problem}

\section{Representable diagonal weakening relation algebras form a discriminator variety}

In this section we define \emph{representable diagonal weakening relation algebras} as those relation algebras where $1$ can be represented as an antichain. Thus in this section when we talk about the concrete binary relation $1$, we mean the diagonal on $X$.
The algebras with this property are the members of $\RwkRA$ that satisfy the identity $1\cdot 0=\bot$.

We show that the simple representable diagonal relation algebras have a discriminator term. A neat consequence is that, unlike representable weakening relation algebras, representable diagonal weakening relation algebras can be defined by an equational theory.

\begin{lemma}\label{discvar}
    For all $R \subseteq X^2$ we have
        $1 \cdot \big(R{;}(R \cdot \simm R)\big) = \bot=
        1 \cdot \big(\simm R{;}(R \cdot \simm R)\big)$.
\end{lemma}
\begin{proof}
    Suppose there exists $(x,x') \in 1 \cdot (R{;}(R \cdot \simm R))$. Because $(x,x') \in 1$ we have $x = x'$. Thus there must exist a $y$ to witness the composition by having $(x,y) \in R, (y,x) \in R \cdot \simm R$. This means that $(x,y) \in R$ and $(y,x) \in \simm R$ and we have reached a contradiction.

    The second equation can be proven by a similar argument or by substitution of $R$ with $\simm R$, the involution law, and the commutativity of meet.
\qed\end{proof}
Let $d_1(R,S) = 1 \cdot \big( R {;}(S \cdot \simm S)\big) \text{ and }
        d_2(R,S) = 1 \cdot \big( \simm S {;}(R \cdot \simm S)\big)$.
\begin{lemma}
    If $R \setminus S\neq\emptyset$ for $R,S \subseteq X^2$ then $d_1(R,S) + d_2(R,S) \neq \bot$.

\end{lemma}

\begin{proof}
    Assume $(x,y)\in R\setminus S$ and consider the two cases, $(y,x) \in S$ and $(y,x) \notin S$. In the first case, because $(x,y) \notin S$ we also have $(y,x) \in \simm S$ and consequently $(y,x) \in S \cdot \simm S$. Hence $(x,x) \in R{;}(S \cdot \simm S)$ and also by definition in $1$ and thus $(x,x) \in d_1(R,S)$.

    In the second case $(y,x) \notin S$ and therefore $(x,y) \in \simm S$. Because $(x,y) \notin S$, $(y,x) \notin \simm S$. By composition $(y,y) \in \simm S;(R \cdot \simm S)$ and by reflexivity of $1$ we also have $(y,y) \in 1 \cdot \big( \simm S {;}(R \cdot \simm S)\big)$.

    In either case we have that at least one of $d_1(R,S), d_2(R,S)$ is nonempty and thus their join is always nonempty given $R\setminus S\ne\emptyset$.
\qed\end{proof}

\begin{theorem}
    Simple diagonal weakening relation algebras have a term $d(a,b,c)$ such that
    $$d(a,b,c) = \begin{cases}
        c & \mbox{if } a = b\\
        a & \mbox{otherwise}
    \end{cases}$$
\end{theorem}

\begin{proof}
    It is easy to see that in simple weakening relation algebras $\top;s;\top=\top$ if $s \neq \bot$ and $\top;s;\top=\bot$ otherwise. By the lemmas above, we have for representable simple algebras that $a = b$ if and only $d_1 + d_2 = \bot$, where $d_i=d_i(a,b)+d_i(b,a)$ for $i=1,2$. Thus $d(a,b,c) = \top;(d_1 + d_2);\top \cdot a + \simm (\top;(d_1 + d_2);\top) \cdot c$ will equal to $c$ if $a=b$ and $a$ otherwise.
\qed\end{proof}

There are several ways to prove the following corollary. We outline the argument that generates a recursively enumerable equational theory.

\begin{corollary}
    Representable diagonal weakening relation algebras form a discriminator variety.
\end{corollary}

\begin{proof}
    The representation game defined for weakening relation algebras only needs an additional move where $\exists$ is requested add $1$ to $\lambda(y,x)$ if $1 \in \lambda(x,y)$ and this game gives rise to a similar style of a recursive axiomatisation as presented in Proposition~\ref{prop:fmlasgame}. If all variables are given unique names, the universal quantifiers can also be moved to the begining of all these formulas. Observe that although these formulas apply to all algebras, the game is played on the homomorphic image of the algebra where $\top$ maps to $\top;a;\top$ where $\sigma_n = s \nleq t \Rightarrow (\phi_n(\N^{1,s,t}) \vee \phi_n(\N^{2,s,t}))$. Thus we can construct a term from any universally quantified first order formula that is equal to $\top;a;\top$ if and only if the formula is true and $\bot$ otherwise. For equations $t = t'$ we take $\top;a;\top \cdot \simm d(t,t',\top;a;\top)$. If a term $t$ corresponds to a formula, then $\simm t \cdot \top;a;\top$ corresponds to its negation and for disjunctions we can take the join of the corresponding terms. Thus every formula $\sigma_n$ has an equivalent equation.
\qed\end{proof}

\section{Representing associative members of $\wkRA_3$ with weakening relations}\label{smallRwkRA}

Sugihara monoids are commutative distributive idempotent involutive residuated lattices. This variety is semilinear, i.e., generated by linearly ordered algebras, and the structure of these algebras is well known.
In particular, the Sugihara monoid $S_n$
is a chain with $n$ elements $\{a_{-k}, a_{-k+1}, \dots, a_{-1}, a_0, a_1, \dots, a_{k-1}, a_k\}$ if $n=2k+1$ is odd, and otherwise for even $n$, $S_n=S_{n+1}\setminus\{a_0\}$. The involution operation is given by $\simm a_i=a_{-i}$ and the
multiplication is $a_i{;} a_j=a_{-\max{|i|,|j|}}$. It follows that in the odd case the identity element is $1=a_0$ and in the even case it is $1=a_1$.

Note that $S_2$ is the 2-element Boolean algebra and that for even $n$, there is a surjective homomorphism from $S_n$ to $S_{n-1}$ that identifies $a_1$ and $a_{-1}$.

It is proved in \cite{Mad2010} that the even Sugihara chains can be represented by algebras of weakening relations. For $S_2$ this is clear since $S_2\cong\text{Rel(1)}$. For $S_4$ an infinite base set is needed with a dense order. E.g., we can take $(\mathbb Q,\le)$ be the poset of rational numbers with the standard order and check that $S_4\cong\{\emptyset,<,\le,\mathbb Q^2\}$ is a representation in Wk$(\mathbb Q,\le)$.







It follows from the consistency of networks that no nontrivial member of $\wkRA_2$ has an element that satisfies $a=\simm a$. Hence 
any finite member of $\wkRA$ has an even number of elements. In particular, the odd Sugihara chains do not have a representation by weakening relations. However they are in the variety generated by all algebras of weakening relations since they are homomorphic images of even Sugihara chains. This shows that $\RWkRA$ is not closed under homomorphic images, so it is a proper quasivariety.


Let $\mathbf{2}=\{0,1\}$ be the two element chain with $0<1$. The algebra $\m{wk}(\m 2)$ is shown in Figure~\ref{6eltwkRA3}, and it has the following six elements: $\emptyset$, $\{(0,1)\}$, $\{(0,0),(0,1)\}$,  $\{(0,1),(1,1)\}$, $\le$, $\m 2\times\m 2$.


The \emph{point algebra} $\m P$ shown in Figure~\ref{pointalg} (see also \cite{Hir1996}) is a representable relation algebra with 3 atoms $id_{\mathbb{Q}},<,>$
where $<$ is the strict order on the rational numbers $\mathbb{Q}$.
It has two weakening subalgebras: $\mathbf{S}_{4}=\{\emptyset,<,\le,\top\}$
and $\mathbf{W}_{6,1}=\{\emptyset,id_{\mathbb{Q}},<,\le,{<}{\cup}{>},\top\}$.
Like the point algebra, both of these algebras can only be represented on an infinite
set.
Note that $\mathbf{W}_{6,1}$ is diagonally representable, while
$\mathbf{S}_{4}$ is not.
\begin{figure}
\begin{center}\scriptsize
\begin{tikzpicture}[scale=.7, every node/.style={circle, draw, fill=white, inner sep=2pt}]
\draw(0,0)--(-1,1)--(-1,0)node[label=left:\protect{\small$id_\mathbb Q$}]{}--(0,-1)--(1,0)node[label=right:\protect{\small$>$}]{}--(1,1)--(0,2)
(1,0)--(0,1)--(-1,0)(0,1)node[label=right:\protect{\small$\!\ge$}]{}--(0,2)
(0,-1)node[label=right:\protect{$\emptyset$}, label=below:$\mathbf P$]{}--
(0,0)node[label=right:\protect{\small$\!<$}]{}--
(1,1)node[label=right:\protect{\small${<}{\cup}{>}$}]{}
(-1,1)node[label=left:\protect{\small$\le$}]{}--
(0,2)node[label=right:\protect{\small${\top=\mathbb Q^2}$}]{};
\end{tikzpicture}
\qquad
\begin{tikzpicture}[scale=.7, every node/.style={circle, draw, fill=white, inner sep=2pt}]
\draw(0,-1)node[label=right:\protect{$\emptyset$}, label=below:$\mathbf S_4$]{}--
(0,0)node[label=right:\protect{\small$<$}]{}--
(0,1)node[label=right:\protect{\small${\le}={\sim}{<}$}]{}--
(0,2)node[label=right:\protect{\small${\top=\mathbb Q^2}$}]{};
\end{tikzpicture}
\qquad
\begin{tikzpicture}[scale=.7, every node/.style={circle, draw, fill=white, inner sep=2pt}]
\draw(0,0)--(-1,1)--(-1,0)node[label=left:\protect{\small$id_\mathbb Q$}]{}--(0,-1)(1,1)--(0,2)
(-1,0)(0,2)
(0,-1)node[label=right:\protect{$\emptyset$}, label=below:$\mathbf W_{6,1}$]{}--
(0,0)node[label=right:\protect{\small${<}={\sim}{\le}$}]{}--
(1,1)node[label=right:\protect{\small${<}{\cup}{>}={\sim}id_\mathbb Q$}]{}
(-1,1)node[label=left:\protect{\small$\le$}]{}--
(0,2)node[label=right:\protect{\small${\top=\mathbb Q^2}$}]{};
\end{tikzpicture}
\end{center}
\caption{The point algebra $\mathbf P$, the weakening subalgebra $\mathbf{S}_{4}$ and the diagonally representable weakening subalgebra $\mathbf{W}_{6,1}$.}\label{pointalg}
\end{figure}
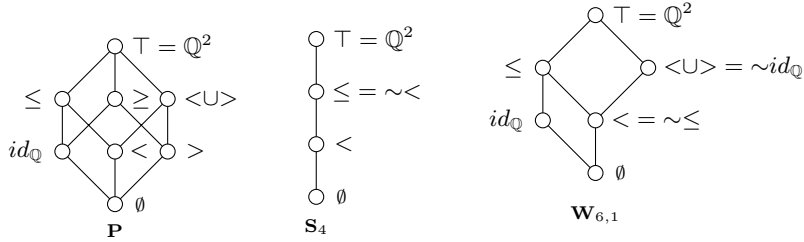

Since $\wkRA_3$ is finitely axiomatised, one can use a model finder such as Mace4 \cite{P9} to compute all members of cardinality $n$ for small values of $n$.  Up to isomorphism there are 14 algebras with 6 elements or fewer in $\wkRA_3$ such that $;$ is associative, shown in Figure~\ref{6eltwkRA3}.
We now briefly describe their representations by weakening relations.

\begin{figure}
\tikzstyle{every picture} = [scale=.5] 
\begin{center}
\begin{tikzpicture}[baseline=0pt]
\node at (0,-1.5)[n]{$\m 1$};
\node(0) at (0,0)[i,label=left:{$0=1$}]{};
\end{tikzpicture}
\ 
\begin{tikzpicture}[baseline=0pt]
\node at (0,-1.5)[n]{$\m 2$};
\node(1) at (0,1)[i,label=right:$1$]{};
\node(0) at (0,0)[i,label=right:$0$]{} edge (1);
\end{tikzpicture}
\ 
\begin{tikzpicture}[baseline=0pt]
\node at (0,-1.5)[n]{$\m 2^2$};
\node(3) at (0,2)[i,label=right:$1$]{};
\node(2) at (1,1)[i,label=left:$\simm a$]{} edge (3);
\node(1) at (-1,1)[i,label=right:$a$]{} edge (3);
\node(0) at (0,0)[i,label=right:$0$]{} edge (1) edge (2);
\end{tikzpicture}
\ 
\begin{tikzpicture}[baseline=0pt]
\node at (0,-1.5)[n]{$\m A_2$};
\node(3) at (0,2)[i]{};
\node(2) at (1,1)[label=left:$0$]{} edge (3);
\node(1) at (-1,1)[i,label=left:$1$][label=right:$0^2_{\ }$]{} edge (3);
\node(0) at (0,0)[i]{} edge (1) edge (2);
\end{tikzpicture}
\ 
\begin{tikzpicture}[baseline=0pt]
\node at (0,-1.5)[n]{$\m A_3$};
\node(3) at (0,2)[i,label=right:$0^2_{\ }$]{};
\node(2) at (1,1)[label=left:$0$]{} edge (3);
\node(1) at (-1,1)[i,label=right:$1$]{} edge (3);
\node(0) at (0,0)[i]{} edge (1) edge (2);
\end{tikzpicture}
\ 
\begin{tikzpicture}[baseline=0pt]
\node at (0,-1.5)[n]{$\m S_{4}$};
\node(3) at (0,3)[i]{};
\node(2) at (0,2)[i,label=right:$1$]{} edge (3);
\node(1) at (0,1)[i,label=right:$0$]{} edge (2);
\node(0) at (0,0)[i]{} edge (1);
\end{tikzpicture}

\begin{tikzpicture}[baseline=0pt]
\node at (0,-1.5)[n]{$\m W_{6,1}$};
\node(5) at (0.5,3)[i,label=right:$0^2_{\ }$]{};
\node(4) at (1.5,2)[label=left:$0$]{} edge (5);
\node(3) at (-0.5,2)[i,label=left:$\simm a$]{} edge (5);
\node(2) at (0.5,1)[i,label=right:$a$]{} edge (3) edge (4);
\node(1) at (-1.5,1)[i,label=right:$1$]{} edge (3);
\node(0) at (-0.5,0)[i]{} edge (1) edge (2);
\end{tikzpicture}
\ 
\begin{tikzpicture}[baseline=0pt]
\node at (0,-1.5)[n]{$\m W_{6,2}$};
\node(5) at (0.5,3)[i,label=right:$0^2_{\ }$]{};
\node(4) at (1.5,2)[label=left:$0$]{} edge (5);
\node(3) at (-0.5,2)[label=left:$\simm a$]{} edge (5);
\node(2) at (0.5,1)[label=right:$a$]{} edge (3) edge (4);
\node(1) at (-1.5,1)[i,label=right:$1$]{} edge (3);
\node(0) at (-0.5,0)[i]{} edge (1) edge (2);
\end{tikzpicture}
\
\begin{tikzpicture}[baseline=0pt]
\node at (0,-1.5)[n]{$\m W_{6,3}$};
\node(5) at (0,4)[i]{};
\node(4) at (0,3)[i,label=right:$1$]{} edge (5);
\node(3) at (1,2)[i,label=left:$\simm a$]{} edge (4);
\node(2) at (-1,2)[i,label=right:$a$]{} edge (4);
\node(1) at (0,1)[i,label=right:$0$]{} edge (2) edge (3);
\node(0) at (0,0)[i]{} edge (1);
\end{tikzpicture}
\
\begin{tikzpicture}[baseline=0pt]
\node at (0,-1.5)[n]{$\m W_{6,4}$};
\node(5) at (0,4)[i]{};
\node(4) at (0,3)[i,label=right:$\simm a$]{} edge (5);
\node(3) at (1,2)[label=left:$0$]{} edge (4);
\node(2) at (-1,2)[i,label=right:$0^2_{\ }$][label=left:$1$]{} edge (4);
\node(1) at (0,1)[i,label=right:$0{;}a$][label=left:$a$]{} edge (2) edge (3);
\node(0) at (0,0)[i]{} edge (1);
\end{tikzpicture}
\
\begin{tikzpicture}[baseline=0pt]
\node at (0,-1.5)[n]{$\m W_{6,5}$};
\node(5) at (0,4)[i]{};
\node(4) at (0,3)[i,label=right:$\simm a$][label=left:$0^2_{\ }$]{} edge (5);
\node(3) at (1,2)[label=left:$0$]{} edge (4);
\node(2) at (-1,2)[i,label=right:$1$]{} edge (4);
\node(1) at (0,1)[i,label=right:$0{;}a$][label=left:$a$]{} edge (2) edge (3);
\node(0) at (0,0)[i]{} edge (1);
\end{tikzpicture}
\
\begin{tikzpicture}[baseline=0pt]
\node at (0,-1.5)[n]{$\m W_{6,6}$};
\node(5) at (0,4)[i,label=right:$0^2_{\ }$]{};
\node(4) at (0,3)[label=right:$\simm a$]{} edge (5);
\node(3) at (1,2)[label=left:$0{;}a$][label=right:$0$]{} edge (4);
\node(2) at (-1,2)[i,label=right:$1$]{} edge (4);
\node(1) at (0,1)[i,label=right:$a$]{} edge (2) edge (3);
\node(0) at (0,0)[i]{} edge (1);
\end{tikzpicture}
\
\begin{tikzpicture}[baseline=0pt]
\node at (0,-1.5)[n]{$\m{wk}(\m 2)$};
\node(5) at (0,4)[i]{};
\node(4) at (0,3)[i,label=right:$1$]{} edge (5);
\node(3) at (1,2)[i,label=left:$\simm a$]{} edge (4);
\node(2) at (-1,2)[i,label=right:$a$]{} edge (4);
\node(1) at (0,1)[label=right:$0$]{} edge (2) edge (3);
\node(0) at (0,0)[i][label=right:$0^2_{\ }$]{} edge (1);
\end{tikzpicture}
\begin{tikzpicture}[baseline=0pt]
\node at (0,-1.5)[n]{$\m S_{6}$};
\node(5) at (0,5)[i]{};
\node(4) at (0,4)[i][i,label=right:$\simm a$]{} edge (5);
\node(3) at (0,3)[i,label=right:$1$]{} edge (4);
\node(2) at (0,2)[i,label=right:$0$]{} edge (3);
\node(1) at (0,1)[i,label=right:$a$]{} edge (2);
\node(0) at (0,0)[i]{} edge (1);
\end{tikzpicture}
\end{center}
\caption{\ All algebras in $\wkRA_4$ up to 6 elements. Black nodes denote idempotent elements ($x{;}x=x$) and $0^2=0{;}0$.}
\label{6eltwkRA3}
\end{figure}
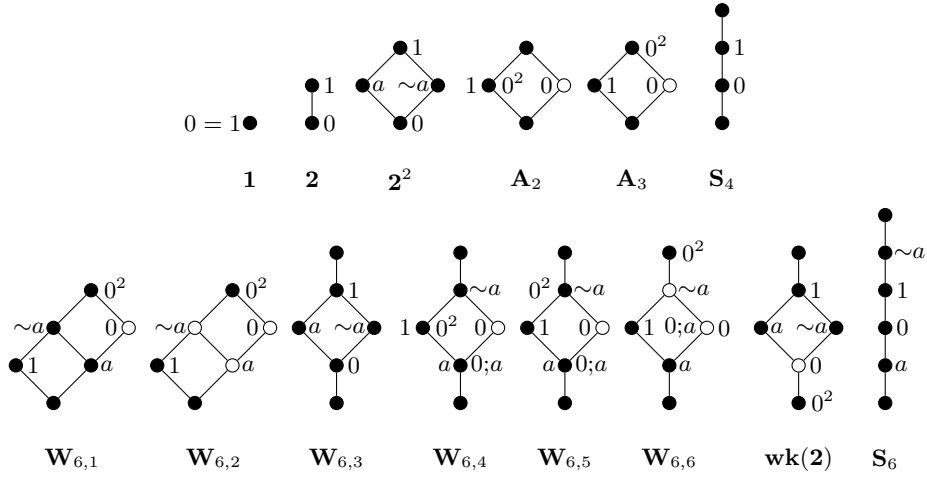

The first 5 are symmetric representable relation algebras, hence they are diagonally representable weakening relation algebras. 

As mentioned above, the Sugihara algebra $\m S_4$ and the algebra $\m W_{6,1}$ are representable as subalgebras of 
the $\sim$-reduct of the point algebra (Figure~\ref{pointalg}). The algebra $\m W_{6,2}$ is representable as $\sim$-subreduct of the complex algebra of $\mathbb Z_7$, where the element $a=\{1,2,4\}$ and $1=\{0\}$.

$\m W_{6,3}$ is subdirectly embedded in a direct product of two copies of $\m S_4$, hence it is representable over the union of two disjoint copies of $\mathbb Q$.

Similarly $\m W_{6,4}$ is represented over $X=(\{0\}\times\mathbb Q)\cup(\{1\}\times\mathbb Q)$ with order $(i,p)\le(j,q)\iff p < q \text{ or } p = q, i = j$. The identity $1$ maps to $\le$ and the element $a$ maps to the relation $\{((i,p),(i,q))\mid i=0,1,\, p<q\}$.

The representation of $\m W_{6,5}$ requires the union of $\{i\}\times\mathbb Q$ for $i\in \{0,1,2\}$. The partial order $\le$ is defined in the same way and $a$ is mapped to the relation $\{((i,p),(i,q))\mid i=0,1,2,\, p<q\}$.

Finally $\m W_{6,6}$ is represented over $X=(\{0\}\times\mathbb Q)\cup(\{1\}\times\mathbb Q)$ with order $(i,p)\le(j,q)\iff i=j \text{ and }p\le q$. The identity $1$ maps to $\le$ and the element $a$ maps to the relation $\{((i,p),(i,q))\mid i=0,1,\, p<q\}$.

We gratefully acknowledge a very useful conversation with Roger Maddux regarding relevance frames, relevance logic and its connections with relation algebras. In particular, formulas (3.101), (3.102) in \cite{Mad2022} provided key insights into the axiomatisation of $\wkRA_3$.

\bibliographystyle{splncs04.bst}
\bibliography{ref}

\end{document}